\documentclass[aps,pra]{revtex4}
\usepackage{amsmath,amssymb,amsthm,graphicx,enumitem,bbm}
\usepackage[
	bookmarks = false,
	colorlinks = true,
	linkcolor = blue,
	urlcolor= black,
	citecolor = blue]{hyperref}
\newtheorem{theorem}{Theorem}
\newtheorem{lemma}{Lemma}

\newtheorem{definition}{Definition}

\newtheorem{example}{Example}
\begin{document}
\title{Locking and unlocking of quantum nonlocality without entanglement \\
in local discrimination of quantum states}
\author{Donghoon Ha}
\affiliation{Department of Applied Mathematics and Institute of Natural Sciences, Kyung Hee University, Yongin 17104, Republic of Korea}
\author{Jeong San Kim}
\email{freddie1@khu.ac.kr}
\affiliation{Department of Applied Mathematics and Institute of Natural Sciences, Kyung Hee University, Yongin 17104, Republic of Korea}
\begin{abstract}
The phenomenon of nonlocality without entanglement(NLWE) arises in discriminating multi-party quantum separable states. Recently, it has been found that the post-measurement information about the prepared subensemble can lock or unlock NLWE in minimum-error discrimination of non-orthogonal separable states. Thus, it is natural to ask whether the availability of the post-measurement information can influence on the occurrence of NLWE even in other state-discrimination strategies. Here, we show that the post-measurement information can be used to lock as well as unlock the occurrence of NLWE in terms of optimal unambiguous discrimination. Our results can provide a useful application for hiding or sharing information based on non-orthogonal separable states.
\end{abstract}
\maketitle
\section*{INTRODUCTION}
\indent Quantum nonlocality is of central importance in multi-party quantum systems.
A typical phenomenon of quantum nonlocality is
quantum entanglement which is a useful resource for multi-party quantum communication\cite{horo2009}.
Quantum entanglement is the correlation that cannot be shared among multiple parties using only \emph{local operations and classical communication}(LOCC)\cite{horo2009,chit20142,chit2019}. 
However, it is also known that some nonlocal phenomena in multi-party quantum systems are still possible even in the absence of quantum entanglement.\\ 
\indent \emph{Nonlocality without entanglement}(NLWE) is another nonlocal phenomenon that arises in discriminating  non-entangled states of multi-party quantum systems\cite{pere1991,benn19991}.
NLWE occurs when what can be achieved with global measurement in discriminating non-entangled states cannot be achieved only by LOCC. In the case of discriminating orthogonal non-entangled states, NLWE occurs when the perfect discrimination cannot be implemented by LOCC\cite{benn19991,divi2003,nise2006,xu20162,hald2019}. On the other hand, in the case of discriminating non-orthogonal non-entangled states, NLWE occurs when the globally optimal discriminations such as minimum-error discrimination\cite{hels1976,hole1979,yuen1975,bae2013} or optimal unambiguous discrimination\cite{ivan1987,pere1988,diek1988,jaeg1995} cannot be implemented by LOCC\cite{chit2013,duan2007,ha20211,ha20212}. 
We also note that some non-local phenomena without entanglement can occur in the generalized probabilistic theories beyond quantum theory\cite{bhat2020}.\\
\indent In quantum state discrimination\cite{chef2000,barn20091,berg2010,bae2015}, orthogonal states can be perfectly discriminated, whereas non-orthogonal states cannot. 
However, some non-orthogonal states 
can be perfectly discriminated when the post-measurement information about the prepared subensemble is available\cite{akib2019}. 
Nevertheless, some non-orthogonal states cannot be perfectly discriminated even the post-measurement information about the prepared subensemble is provided\cite{ball2008,gopa2010,carm2018}.
Therefore, in optimal discriminations with the post-measurement information about the prepared subensemble, the NLWE phenomenon arises
when the globally optimal discrimination cannot be implemented by LOCC with the help of post-measurement information.
Recently, it was shown that the availability of post-measurement information can lock or unlock NLWE in terms of minimum-error discrimination\cite{ha20213}, therefore it is natural to ask whether the post-measurement information affects the occurrence of NLWE in terms of state-discrimination strategies other than minimum-error discrimination.\\ 
\indent Here, we show that even in optimal unambiguous discrimination,
the availability of the post-measurement information about the prepared subensemble can affect the occurrence of NLWE.
We first provide an ensemble of two-qubit product states having NLWE in terms of optimal unambiguous discrimination, and show that the availability of post-measurement information about the prepared subensemble vanishes the occurrence of NLWE, therefore \emph{locking NLWE in terms of optimal unambiguous discrimination by post-measurement information}. We further provide another ensemble of two-qubit product state that does not have NLWE in terms of optimal unambiguous discrimination, and show that NLWE in the optimal unambiguous discrimination can be released when the post-measurement information about the prepared subensemble is provided. Thus \emph{unlocking NLWE in terms of optimal unambiguous discrimination by post-measurement information}.\\
\indent This paper is organized as follows. 
First, we present the form of two-qubit product state ensemble to be considered.
In the ``\hyperref[mtdsec]{Methods}'' Section,
we review the definitions and properties
with respect to optimal unambiguous discrimination without and with post-measurement information and provide some useful lemmas in optimal local discrimination. As a main result of this paper, 
we provide a quantum state ensemble
consisting of four two-qubit product states
and show the occurrence of NLWE in terms of optimal unambiguous discrimination.
With the same ensemble, we further show that 
NLWE does not occur in the optimal unambiguous discrimination with the post-measurement information about the prepared subensemble is available.
As another main result of this paper, 
we provide another quantum state ensemble consisting of four two-qubit product states and show the non-occurrence of NLWE in terms of optimal unambiguous discrimination. With the same ensemble, we further show that NLWE occurs in the optimal unambiguous discrimination with the post-measurement information about the prepared subensemble.
\section*{RESULTS}
\indent Throughout this paper, we only consider
the situation 
of unambiguously discriminating four states from the quantum state ensemble, 
\begin{equation}\label{eq:ense}
\mathcal{E}=\{\eta_{i},\rho_{i}\}_{i\in\Lambda},\ \Lambda=\{0,1,+,-\},
\end{equation}
where $\rho_{i}$ is a $2\otimes2$ non-entangled pure state, 
\begin{equation}\label{eq:varphis}
\rho_{i}=|\varphi_{i}\rangle\!\langle\varphi_{i}|
\ \ \mbox{for each $i\in\Lambda$},
\end{equation}
and $\{|\varphi_{i}\rangle\}_{i\in\Lambda}$
is a product basis of $\mathcal{H}$.
Each $\eta_{i}$ is the probability that
the state $\rho_{i}$ is prepared.\\
\indent The ensemble $\mathcal{E}$ can be seen as an ensemble consisting of 
two subensembles,
\begin{equation}\label{eq:suben}
\begin{array}{ll}
\mathcal{E}_{0}=\{\eta_{i}/\sum_{j\in\mathsf{A}_{0}}\eta_{j},\rho_{i}\}_{i\in\mathsf{A}_{0}},& 
\mathsf{A}_{0}=\{\,0\,,\,1\,\},\\[1mm]
\mathcal{E}_{1}=\{\eta_{i}/\sum_{j\in\mathsf{A}_{1}}\eta_{j},\rho_{i}\}_{i\in\mathsf{A}_{1}},&
\mathsf{A}_{1}=\{+,-\},
\end{array}
\end{equation}
where $\mathcal{E}_{0}$ and $\mathcal{E}_{1}$ are prepared with probabilities $\sum_{j\in\mathsf{A}_{0}}\eta_{j}$ and $\sum_{j\in\mathsf{A}_{1}}\eta_{j}$, respectively.
The definitions and properties related to optimal unambiguous discrimination of $\mathcal{E}$ without and with post-measurement information are provided in the ``\hyperref[mtdsec]{Methods}'' Section.\\
\indent Before we deliver our main results in the following subsections, we first provide the concepts of NLWE, NLWE with post-measurement information, and locking/unlocking NLWE by post-measurement information.
\begin{definition}\label{def:nlwe}
For optimal unambiguous discrimination of a separable ensemble $\mathcal{E}$ in Eq.~\eqref{eq:ense}, NLWE occurs if and only if optimal unambiguous discrimination of $\mathcal{E}$ cannot be realized only by LOCC measurements, that is,
\begin{equation}\label{eq:dnlwe}
p_{\rm L}(\mathcal{E})<p_{\rm G}(\mathcal{E}).
\end{equation}
\end{definition}
\indent In discriminating orthogonal non-entangled states, NLWE  occurs when the perfect discrimination cannot be realized by LOCC.
Thus, the NLWE phenomenon of orthogonal non-entangled states is a special case of the NLWE phenomenon defined in Definition~\ref{def:nlwe}, that is, $p_{\rm L}(\mathcal{E})<p_{\rm G}(\mathcal{E})=1$. In the following definition, we provide the concept of NLWE in optimal unambiguous discrimination of $\mathcal{E}$ when the post-measurement information about the prepared subensemble is available. 
\begin{definition}
For optimal unambiguous discrimination of a separable ensemble $\mathcal{E}$ in Eq.~\eqref{eq:ense} with post-measurement information $b\in\{0,1\}$ about the prepared subensemble $\mathcal{E}_{b}$ in Eq.~\eqref{eq:suben},
NLWE occurs if and only if optimal unambiguous discrimination of $\mathcal{E}$ with post-measurement information cannot be realized only by LOCC measurements, that is,
\begin{equation}
p_{\rm L}^{\rm PI}(\mathcal{E})<p_{\rm G}^{\rm PI}(\mathcal{E}).
\end{equation}
\end{definition}
\indent Now, we provide the concepts of \emph{locking} and \emph{unlocking} NLWE by post-measurement information in optimal unambiguous discrimination of $\mathcal{E}$.
\begin{definition}
Let us consider the optimal unambiguous discrimination of a separable ensemble $\mathcal{E}$ in Eq.~\eqref{eq:ense}. 
\begin{enumerate}[leftmargin=4.6mm]
\item[1.] The post-measurement information $b\in\{0,1\}$ about the prepared subensemble $\mathcal{E}_{b}$ in Eq.~\eqref{eq:suben} \emph{locks} NLWE if NLWE occurs in discriminating the states of $\mathcal{E}$,
\begin{equation}
p_{\rm L}(\mathcal{E})<p_{\rm G}(\mathcal{E}),
\end{equation}
whereas
NLWE does not occur when the post-measurement information $b$ about the prepared subensemble is available, 
\begin{equation}
p_{\rm L}^{\rm PI}(\mathcal{E})=p_{\rm G}^{\rm PI}(\mathcal{E}).
\end{equation}
\item[2.] The post-measurement information $b$ about the prepared subensemble $\mathcal{E}_{b}$ \emph{unlocks} NLWE if NLWE does not occur in discriminating the states of $\mathcal{E}$,
\begin{equation}
p_{\rm L}(\mathcal{E})=p_{\rm G}(\mathcal{E}),
\end{equation} 
whereas NLWE occurs when the post-measurement information $b$ about the prepared subensemble is available, 
\begin{equation}
p_{\rm L}^{\rm PI}(\mathcal{E})<p_{\rm G}^{\rm PI}(\mathcal{E}).
\end{equation}
\end{enumerate}
\end{definition}

\subsection*{Locking NLWE by post-measurement information in optimal unambiguous discrimination}
In this section, we consider a situation where the post-measurement information about the prepared subensemble $\mathcal{E}_{b}$ \emph{locks} NLWE in terms of optimal unambiguous discrimination.
We first provide a specific example of 
a state ensemble $\mathcal{E}$ and show that NLWE in terms of optimal unambiguous discrimination occurs. With the same ensemble, we further show that the occurrence of NLWE in terms of optimal unambiguous discrimination can be vanished when post-measurement information is provided, thus locking NLWE by post-measurement information.
\begin{example}[\cite{ha20213}]\label{ex:lock}
Let us consider the ensemble $\mathcal{E}$ in Eq.~\eqref{eq:ense} with
\begin{equation}\label{eq:ftqs02}
\begin{array}{lcllcllcl}
\eta_{0}&=&\frac{\gamma}{2(1+\gamma)},&
\rho_{0}&=&|\varphi_{0}\rangle\!\langle\varphi_{0}|,&
|\varphi_{0}\rangle&=&
|0\rangle\otimes|0\rangle,\\[1mm]
\eta_{1}&=&\frac{\gamma}{2(1+\gamma)},&
\rho_{1}&=&|\varphi_{1}\rangle\!\langle\varphi_{1}|,&
|\varphi_{1}\rangle&=&
|0\rangle\otimes|1\rangle,\\[1mm]
\eta_{+}&=&\frac{1}{2(1+\gamma)},& 
\rho_{+}&=&|\varphi_{+}\rangle\!\langle\varphi_{+}|,&
|\varphi_{+}\rangle&=&
|+\rangle\otimes|+\rangle,\\[1mm]
\eta_{-}&=&\frac{1}{2(1+\gamma)},& 
\rho_{-}&=&|\varphi_{-}\rangle\!\langle\varphi_{-}|,&
|\varphi_{-}\rangle&=&
|-\rangle\otimes|-\rangle,
\end{array}
\end{equation}
where $2\leqslant\gamma<\infty$,  $\{|0\rangle,|1\rangle\}$ is the standard basis in one-qubit system, and
$|\pm\rangle=\frac{1}{\sqrt{2}}(|0\rangle\pm|1\rangle)$.
In this case, the subensembles in Eq.~\eqref{eq:suben} become
\begin{equation}
\begin{array}{l}
\mathcal{E}_{0}=\{\frac{1}{2},|0\rangle\!\langle0|\!\otimes\!|0\rangle\!\langle0|,\ 
\frac{1}{2},|0\rangle\!\langle0|\!\otimes\!|1\rangle\!\langle1|\},\\[1mm]
\mathcal{E}_{1}=\{\frac{1}{2},|+\rangle\!\langle+|\!\otimes\!|+\rangle\!\langle+|,\ 
\frac{1}{2},|-\rangle\!\langle-|\!\otimes\!|-\rangle\!\langle-|\},
\end{array}
\end{equation}
with the probabilities of preparation  
$\frac{\gamma}{1+\gamma}$ and $\frac{1}{1+\gamma}$, respectively.
\end{example}
\indent To show the occurrence of NLWE in terms of optimal unambiguous discrimination about the ensemble $\mathcal{E}$ in Example~\ref{ex:lock}, we first evaluate
the optimal success probability $p_{\rm G}(\mathcal{E})$ defined in Eq.~\eqref{eq:pgeoud} of the ``\hyperref[mtdsec]{Methods}'' Section.
The reciprocal vectors $\{|\tilde{\varphi}_{i}\rangle\}_{i\in\Lambda}$ corresponding to
$\{|\varphi_{i}\rangle\}_{i\in\Lambda}$ defined in Eq.~\eqref{eq:ftqs02} are 
\begin{equation}
\begin{array}{l}
|\tilde{\varphi}_{0}\rangle=\sqrt{2}|\Phi_{-}\rangle,\ 
|\tilde{\varphi}_{+}\rangle=\sqrt{2}|1+\rangle,\\[1mm]
|\tilde{\varphi}_{1}\rangle=\sqrt{2}|\Psi_{-}\rangle,\ 
|\tilde{\varphi}_{-}\rangle=-\sqrt{2}|1-\rangle,
\end{array}
\end{equation}
where
\begin{equation}
\begin{array}{c}
|\Phi_{\pm}\rangle=\frac{1}{\sqrt{2}}|00\rangle\pm\frac{1}{\sqrt{2}}|11\rangle,\
|\Psi_{\pm}\rangle=\frac{1}{\sqrt{2}}|01\rangle\pm\frac{1}{\sqrt{2}}|10\rangle.
\end{array}
\end{equation}
We can easily verify that the following $\{M_{i}\}_{i\in\overline{\Lambda}}$ is an unambiguous measurement
satisfying the error-free condition in Eq.~\eqref{eq:udm}: 
\begin{equation}\label{eq:oudm1}
\begin{array}{lcl}
M_{0}=|\Phi_{-}\rangle\!\langle\Phi_{-}|,\
M_{+}=0,\\[1mm]
M_{1}=|\Psi_{-}\rangle\!\langle\Psi_{-}|,\
M_{-}=0,\\[1mm]
M_{?}=\mathbbm{1}-|\Phi_{+}\rangle\!\langle\Phi_{+}|
-|\Psi_{+}\rangle\!\langle\Psi_{+}|.
\end{array}
\end{equation}
Also, it is optimal because Condition~\eqref{eq:kkt} holds for this unambiguous measurement along with a positive-semidefinite operator
\begin{equation}
\begin{array}{c}
K=\frac{\gamma}{4(1+\gamma)}(|\Phi_{-}\rangle\!\langle\Phi_{-}|
+|\Psi_{-}\rangle\!\langle\Psi_{-}|).
\end{array}
\end{equation}
Thus, the optimality of the measurement $\{M_{i}\}_{i\in\overline{\Lambda}}$ in Eq.~\eqref{eq:oudm1} and the definition of $p_{\rm G}(\mathcal{E})$ lead us to
\begin{equation}\label{eq:qgee0}
\begin{array}{c}
p_{\rm G}(\mathcal{E})=\mathrm{Tr}K=\frac{\gamma}{2(1+\gamma)}=\eta_{0}.
\end{array}
\end{equation}
\indent 
In order to obtain the maximum success probability $p_{\rm L}(\mathcal{E})$ defined in Eq.~\eqref{eq:qlelocc} of the ``\hyperref[mtdsec]{Methods}'' Section, we consider lower and upper bounds of $p_{\rm L}(\mathcal{E})$.
A lower bound of $p_{\rm L}(\mathcal{E})$ can be obtained from the following unambiguous measurement $\{M_{i}\}_{i\in\overline{\Lambda}}$,
\begin{equation}\label{eq:loccud2}
\begin{array}{lcl}
M_{0}=0,\ M_{+}=|1\rangle\!\langle1|\otimes|+\rangle\!\langle+|,\\[1mm]
M_{1}=0,\ M_{-}=|1\rangle\!\langle1|\otimes|-\rangle\!\langle-|,\\[1mm]
M_{?}=|0\rangle\!\langle0|\otimes(|+\rangle\!\langle+|+|-\rangle\!\langle-|),
\end{array}
\end{equation}
which can be implemented by finite-round LOCC because it can be realized by performing local measurements $\{|0\rangle\!\langle0|,|1\rangle\!\langle1|\}$ and $\{|+\rangle\!\langle+|,|-\rangle\!\langle-|\}$ 
on first and second subsystems, respectively.
As we can easily verify that the success probability for the unambiguous LOCC measurement in Eq.~\eqref{eq:loccud2} is $\frac{1}{2(1+\gamma)}$, the success probability is obviously a lower bound of $p_{\rm L}(\mathcal{E})$,
\begin{equation}\label{eq:qlelwb}
\begin{array}{c}
p_{\rm L}(\mathcal{E})\geqslant
\frac{1}{2(1+\gamma)}=\eta_{+}.
\end{array}
\end{equation}
\indent To obtain an upper bound of $p_{\rm L}(\mathcal{E})$, 
let us consider a positive-semidefinite operator
\begin{equation}
\begin{array}{c}
H=\frac{1}{4(1+\gamma)}|1\rangle\!\langle1|\otimes(|+\rangle\!\langle+|+|-\rangle\!\langle-|)
\end{array}
\end{equation}
with
\begin{equation}
\begin{array}{c}
\langle\tilde{\varphi}_{+}|H|\tilde{\varphi}_{+}\rangle
=\eta_{+}=\eta_{-}=\langle\tilde{\varphi}_{-}|H|\tilde{\varphi}_{-}\rangle.
\end{array}
\end{equation}
Lemma~\ref{lem:qletrh} in the ``\hyperref[mtdsec]{Methods}'' Section leads us to
\begin{equation}\label{eq:qleupb}
\begin{array}{c}
p_{\rm L}(\mathcal{E})\leqslant\mathrm{Tr}H=\frac{1}{2(1+\gamma)}=\eta_{+}.
\end{array}
\end{equation}
Inequalities \eqref{eq:qlelwb} and \eqref{eq:qleupb} imply
\begin{equation}\label{eq:qleep}
p_{\rm L}(\mathcal{E})=\eta_{+}.
\end{equation}
\indent From Eqs.~\eqref{eq:qgee0} and \eqref{eq:qleep}, we note that there exists
a nonzero gap between
$p_{\rm G}(\mathcal{E})$ and $p_{\rm L}(\mathcal{E})$,
\begin{equation}\label{eq:plleqpg}
p_{\rm L}(\mathcal{E})=\eta_{+}<\eta_{0}=p_{\rm G}(\mathcal{E}),
\end{equation}
thus NLWE occurs in terms of optimal unambiguous discrimination in discriminating the states of the ensemble $\mathcal{E}$ in Example~\ref{ex:lock}.\\
\indent Now, we show that the availability of post-measurement information about the prepared subensemble vanishes the occurrence of NLWE in Inequality~\eqref{eq:plleqpg}. 
To show it, we use the fact that the states of $\mathcal{E}$ in Example~\ref{ex:lock} can be unambiguously discriminated without inconclusive results using LOCC when the post-measurement information about the prepared subensemble is available\cite{ha20213}, or equivalently,
\begin{equation}
p_{\rm L}^{\rm PI}(\mathcal{E})\geqslant1.
\end{equation}
From the definitions of $p_{\rm L}^{\rm PI}(\mathcal{E})$ and $p_{\rm G}^{\rm PI}(\mathcal{E})$, we note that 
\begin{equation}
p_{\rm G}^{\rm PI}(\mathcal{E})\geqslant
p_{\rm L}^{\rm PI}(\mathcal{E}).
\end{equation}
As both $p_{\rm G}^{\rm PI}(\mathcal{E})$ and $p_{\rm L}^{\rm PI}(\mathcal{E})$ are bound above by 1, we have
\begin{equation}\label{eq:pleqpg}
p_{\rm L}^{\rm PI}(\mathcal{E})=p_{\rm G}^{\rm PI}(\mathcal{E})=1.
\end{equation}
Thus, NLWE does not occur in terms of optimal unambiguous discrimination in discriminating the states of the ensemble $\mathcal{E}$ in Example~\ref{ex:lock} when the post-measurement information about the prepared subensemble is available.\\
\begin{figure}[!t]
\centerline{\includegraphics*[bb=20 20 430 310,scale=0.8]{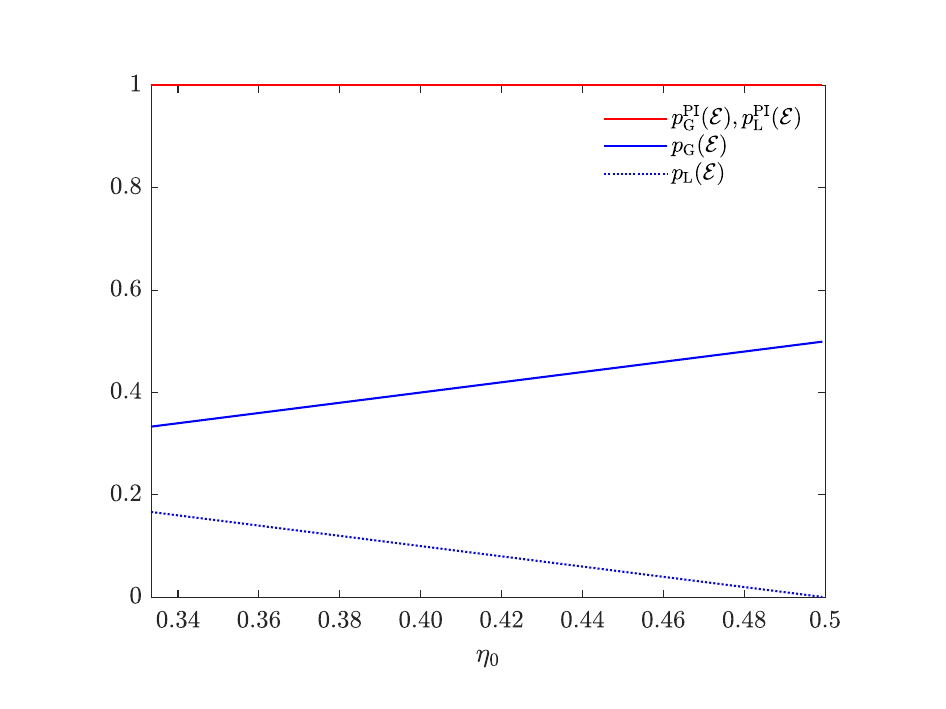}}
\caption{{\bf Locking NLWE by post-measurement information in terms of optimal unambiguous discrimination.} For all $\eta_{0}\in[\frac{1}{3},\frac{1}{2})$, 
$p_{\rm L}(\mathcal{E})$(dashed blue) is less than
$p_{\rm G}(\mathcal{E})$(solid blue), 
but $p_{\rm L}^{\rm PI}(\mathcal{E})$(red)
is equal to $p_{\rm G}^{\rm PI}(\mathcal{E})$(red).}\label{fig:loud}
\end{figure}
\indent Inequality \eqref{eq:plleqpg}
shows that NLWE occurs in terms of optimal unambiguous discrimination about the ensemble $\mathcal{E}$ in Example~\ref{ex:lock}, whereas Eq.~\eqref{eq:pleqpg} shows that NLWE does not occur when post-measurement information is available.
Figure~\ref{fig:loud} illustrates the relative order of $p_{\rm G}(\mathcal{E})$, $p_{\rm L}(\mathcal{E})$, $p_{\rm G}^{\rm PI}(\mathcal{E})$, and $p_{\rm L}^{\rm PI}(\mathcal{E})$ for the range of $\frac{1}{3}\leqslant\eta_{0}<\frac{1}{2}$.
\begin{theorem}\label{thm:loud}
For optimal unambiguous discrimination of the ensemble $\mathcal{E}$ in Example~\ref{ex:lock},
the post-measurement information about the prepared subensemble locks NLWE.
\end{theorem}
\subsection*{Unlocking NLWE by post-measurement information in optimal unambiguous discrimination}
In this section, we consider the opposite situation to the previous section; the post-measurement information about the prepared subensemble $\mathcal{E}_{b}$ in Eq.~\eqref{eq:suben} \emph{unlocks} NLWE.
After providing an example of a state ensemble $\mathcal{E}$, we first show that NLWE in terms of optimal unambiguous discrimination does not occur in discriminating the states of the ensemble. With the same ensemble, we further show the occurrence of NLWE in terms of optimal unambiguous discrimination in the state discrimination with the help of post-measurement information, thus unlocking NLWE by post-measurement information.
\begin{example}[\cite{ha20213}]\label{ex:unlock}
Let us consider the ensemble $\mathcal{E}$ in Eq.~\eqref{eq:ense} with
\begin{equation}\label{eq:ftqs01}
\begin{array}{lcllcllcl}
\eta_{0}&=&\frac{\gamma}{2(1+\gamma)},&
\rho_{0}&=&|\varphi_{0}\rangle\!\langle\varphi_{0}|,&
|\varphi_{0}\rangle&=&
|0\rangle\otimes|0\rangle,\\[1mm]
\eta_{1}&=&\frac{\gamma}{2(1+\gamma)},&
\rho_{1}&=&|\varphi_{1}\rangle\!\langle\varphi_{1}|,&
|\varphi_{1}\rangle&=&
|0\rangle\otimes|1\rangle,\\[1mm]
\eta_{+}&=&\frac{1}{2(1+\gamma)},& 
\rho_{+}&=&|\varphi_{+}\rangle\!\langle\varphi_{+}|,&
|\varphi_{+}\rangle&=&
|+\rangle\otimes|+\rangle,\\[1mm]
\eta_{-}&=&\frac{1}{2(1+\gamma)},& 
\rho_{-}&=&|\varphi_{-}\rangle\!\langle\varphi_{-}|,&
|\varphi_{-}\rangle&=&
|+\rangle\otimes|-\rangle,
\end{array}
\end{equation}
where $2\leqslant\gamma<\infty$.
In this case, the subensembles in Eq.~\eqref{eq:suben} become
\begin{equation}
\begin{array}{l}
\mathcal{E}_{0}=\{\frac{1}{2},|0\rangle\!\langle0|\!\otimes\!|0\rangle\!\langle0|,\ 
\frac{1}{2},|0\rangle\!\langle0|\!\otimes\!|1\rangle\!\langle1|\},\\[1mm]
\mathcal{E}_{1}=\{\frac{1}{2},|+\rangle\!\langle+|\!\otimes\!|+\rangle\!\langle+|,\ 
\frac{1}{2},|+\rangle\!\langle+|\!\otimes\!|-\rangle\!\langle-|\},
\end{array}
\end{equation}
with the probabilities of preparation  
$\frac{\gamma}{1+\gamma}$ and $\frac{1}{1+\gamma}$, respectively.
\end{example}
\indent To show the non-occurrence of NLWE in terms of optimal unambiguous discrimination about the ensemble $\mathcal{E}$ in Example~\ref{ex:unlock}, we first evaluate the optimal success probability $p_{\rm G}(\mathcal{E})$ defined in Eq.~\eqref{eq:pgeoud} of the ``\hyperref[mtdsec]{Methods}'' Section.
Since the reciprocal vectors $\{|\tilde{\varphi}_{i}\rangle\}_{i\in\Lambda}$ corresponding to
$\{|\varphi_{i}\rangle\}_{i\in\Lambda}$ defined in Eq.~\eqref{eq:ftqs01} are 
\begin{equation}\label{eq:tpsiv}
\begin{array}{l}
|\tilde{\varphi}_{0}\rangle=\sqrt{2}|-\rangle\otimes|0\rangle,\ 
|\tilde{\varphi}_{+}\rangle=\sqrt{2}|1\rangle\otimes|+\rangle,\\[1mm]
|\tilde{\varphi}_{1}\rangle=\sqrt{2}|-\rangle\otimes|1\rangle,\ 
|\tilde{\varphi}_{-}\rangle=\sqrt{2}|1\rangle\otimes|-\rangle,
\end{array}
\end{equation}
the following measurement $\{M_{i}\}_{i\in\overline{\Lambda}}$ satisfies the error-free condition in Eq.~\eqref{eq:udm}, 
\begin{equation}\label{eq:oudm}
\begin{array}{lcl}
M_{0}=|-\rangle\!\langle-|\otimes|0\rangle\!\langle0|,\
M_{+}=0,\\[1mm]
M_{1}=|-\rangle\!\langle-|\otimes|1\rangle\!\langle1|,\
M_{-}=0,\\[1mm]
M_{?}=|+\rangle\!\langle+|\otimes(|0\rangle\!\langle0|+|1\rangle\!\langle1|).
\end{array}
\end{equation}
Moreover, the unambiguous measurement is optimal because Condition~\eqref{eq:kkt} holds for
this unambiguous measurement along with the following positive-semidefinite operator
\begin{equation}
\begin{array}{c}
K=\frac{\gamma}{4(1+\gamma)}|-\rangle\!\langle-|\otimes(|0\rangle\!\langle0|+|1\rangle\!\langle1|).
\end{array}
\end{equation}
Thus, the optimality of the measurement $\{M_{i}\}_{i\in\overline{\Lambda}}$ in Eq.~\eqref{eq:oudm} and the definition of $p_{\rm G}(\mathcal{E})$ lead us to
\begin{equation}
\begin{array}{c}
p_{\rm G}(\mathcal{E})=\mathrm{Tr}K=\frac{\gamma}{2(1+\gamma)}=\eta_{0}.
\end{array}
\end{equation}
\indent The measurement given in Eq.~\eqref{eq:oudm} can be performed using finite-round LOCC;
two local measurements
$\{|+\rangle\!\langle+|,|-\rangle\!\langle-|\}$ and 
$\{|0\rangle\!\langle0|,|1\rangle\!\langle1|\}$
are performed on first and second subsystems, respectively.
Thus, the success probability for the unambiguous LOCC measurement in Eq.~\eqref{eq:oudm} is a lower bound of $p_{\rm L}(\mathcal{E})$ defined in Eq.~\eqref{eq:qlelocc}, therefore
\begin{equation}\label{eq:qlee0}
\begin{array}{c}
p_{\rm L}(\mathcal{E})\geqslant\eta_{0},
\end{array}
\end{equation}
Moreover, from the definition of $p_{\rm G}(\mathcal{E})$ and $p_{\rm L}(\mathcal{E})$ in Eqs.~\eqref{eq:pgeoud} and \eqref{eq:qlelocc}, respectively, we have
\begin{equation}\label{eq:pgeple}
p_{\rm G}(\mathcal{E})\geqslant p_{\rm L}(\mathcal{E}).
\end{equation}
Inequalities \eqref{eq:qlee0} and \eqref{eq:pgeple} lead us to
\begin{equation}\label{eq:eqplpg}
p_{\rm L}(\mathcal{E})=p_{\rm G}(\mathcal{E})=\eta_{0}.
\end{equation}
Thus, NLWE does not occur in terms of optimal unambiguous discrimination in discriminating the states of the ensemble $\mathcal{E}$ in Example~\ref{ex:unlock}.\\
\indent Now, we show that NLWE in terms of optimal unambiguous discrimination occurs when the post-measurement information about the prepared subensemble is available. 
To show it, we use the fact that the states of $\mathcal{E}$ in Example~\ref{ex:unlock} can be unambiguously discriminated without inconclusive results when the post-measurement information about the prepared subensemble is available\cite{ha20213}, or equivalently,
\begin{equation}
p_{\rm G}^{\rm PI}(\mathcal{E})\geqslant1.
\end{equation}
As $p_{\rm G}^{\rm PI}(\mathcal{E})$ is bound above by 1, we have
\begin{equation}\label{eq:qgpie1}
p_{\rm G}^{\rm PI}(\mathcal{E})=1.
\end{equation}
\indent To obtain the maximum success probability $p_{\rm L}^{\rm PI}(\mathcal{E})$ in Eq.~\eqref{eq:qlpie} of the ``\hyperref[mtdsec]{Methods}'' Section, we consider lower and upper bounds of $p_{\rm L}^{\rm PI}(\mathcal{E})$.
For a lower bound of $p_{\rm L}^{\rm PI}(\mathcal{E})$, let us first consider 
the following measurement $\{M_{\vec{\omega}}\}_{\vec{\omega}\in\Omega}$,
\begin{equation}\label{eq:fmem}
\begin{array}{l}
M_{(0,?)}=|\nu_{-}\rangle\!\langle\nu_{-}|\otimes|0\rangle\!\langle0|,\
M_{(?,+)}=|\nu_{+}\rangle\!\langle\nu_{+}|\otimes|+\rangle\!\langle+|,\\[1mm]
M_{(1,?)}=|\nu_{-}\rangle\!\langle\nu_{-}|\otimes|1\rangle\!\langle1|,\
M_{(?,-)}=|\nu_{+}\rangle\!\langle\nu_{+}|\otimes|-\rangle\!\langle-|,\\[1mm]
M_{\vec{\omega}}=0\ \forall\vec{\omega}\in\{(0,+),(0,-),(1,+),(1,-),(?,?)\},\\
\end{array}
\end{equation}
where
\begin{equation}\label{eq:nupm}
\begin{array}{c}
|\nu_{\pm}\rangle=\sqrt{\frac{1}{2}\mp\frac{\gamma}{2\sqrt{1+\gamma^{2}}}}|0\rangle\pm\sqrt{\frac{1}{2}\pm\frac{\gamma}{2\sqrt{1+\gamma^{2}}}}|1\rangle.
\end{array}
\end{equation}
\indent The measurement given in Eq.~\eqref{eq:fmem} 
is unambiguous because it satisfies the error-free condition in Eq.~\eqref{eq:udmpi}.
Moreover, this measurement can be performed with finite-round LOCC;
we first measure $\{|\nu_{+}\rangle\!\langle\nu_{+}|,|\nu_{-}\rangle\!\langle\nu_{-}|\}$ on first subsystem, and then measure $\{|+\rangle\!\langle+|,|-\rangle\!\langle-|\}$ or $\{|0\rangle\!\langle0|,|1\rangle\!\langle1|\}$ on second subsystem depending on the first measurement result $|\nu_{+}\rangle\!\langle\nu_{+}|$ or $|\nu_{-}\rangle\!\langle\nu_{-}|$. 
As we can verify from a straightforward calculation that the success probability for the unambiguous LOCC measurement in Eq.~\eqref{eq:fmem} is 
\begin{equation}
\sum_{b\in\{0,1\}}\sum_{i\in\mathsf{A}_{b}}
\eta_{i}\mathrm{Tr}\Big[\rho_{i} \sum_{\substack{\vec{\omega}\in\Omega\\\omega_{b}=i}} M_{\vec{\omega}}\Big]=
\sum_{i\in\mathsf{A}_{0}}\eta_{i}\mathrm{Tr}(\rho_{i}M_{(i,?)})+\sum_{j\in\mathsf{A}_{1}}\eta_{j}\mathrm{Tr}(\rho_{j}M_{(?,j)})=\frac{1}{2}\Big(1+\frac{\sqrt{1+\gamma^{2}}}{1+\gamma}\Big),
\end{equation}
thus the definition of $p_{\rm L}^{\rm PI}(\mathcal{E})$ lead us to
\begin{equation}\label{eq:qlpiegeq}
\begin{array}{c}
p_{\rm L}^{\rm PI}(\mathcal{E})\geqslant \frac{1}{2}\Big(1+\frac{\sqrt{1+\gamma^{2}}}{1+\gamma}\Big).
\end{array}
\end{equation}
We also note that the measurement in Eq.~\eqref{eq:fmem} yields $p_{\rm guess}(\mathcal{E})$ defined in Eq.~\eqref{eq:plelocc} of the ``\hyperref[mtdsec]{Methods}'' Section when considering $M_{0}=M_{(0,?)}$, $M_{1}=M_{(1,?)}$, $M_{+}=M_{(?,+)}$, and $M_{-}=M_{(?,-)}$\cite{ha20213}, that is,
\begin{equation}
\begin{array}{c}
p_{\rm guess}(\mathcal{E})=\frac{1}{2}\Big(1+\frac{\sqrt{1+\gamma^{2}}}{1+\gamma}\Big).
\end{array}
\end{equation}
\indent In order to obtain an upper bound of $p_{\rm L}^{\rm PI}(\mathcal{E})$, let us consider
the assumption of Lemma~\ref{lem:qlpieple} in the ``\hyperref[mtdsec]{Methods}'' Section.
For each $(\omega_{0},\omega_{1})\in\mathsf{A}_{0}\times\mathsf{A}_{1}$, there does not exist any nonzero product vector
$|v\rangle=|a\rangle\otimes|b\rangle$ satisfying Condition~\eqref{eq:npvc}; otherwise, $|a\rangle$ is not orthogonal to both $|0\rangle$ and $|+\rangle$. At the same time, $|b\rangle$ is orthogonal to the $|k\rangle$'s with $k\in\Lambda\setminus\{\omega_{0},\omega_{1}\}$, which leads us a contradiction.
Thus, the guessing probability of $\mathcal{E}$ is also an upper bound of $p_{\rm L}^{\rm PI}(\mathcal{E})$ due to Lemma~\ref{lem:qlpieple} in the ``\hyperref[mtdsec]{Methods}'' Section, that is,
\begin{equation}\label{eq:qlpiepge}
\begin{array}{c}
p_{\rm L}^{\rm PI}(\mathcal{E})\leqslant p_{\rm guess}(\mathcal{E})=\frac{1}{2}\Big(1+\frac{\sqrt{1+\gamma^{2}}}{1+\gamma}\Big).
\end{array}
\end{equation}
Inequalities \eqref{eq:qlpiegeq} and \eqref{eq:qlpiepge} imply
\begin{equation}\label{eq:gappi}
\begin{array}{c}
p_{\rm L}^{\rm PI}(\mathcal{E})=\frac{1}{2}\Big(1+\frac{\sqrt{1+\gamma^{2}}}{1+\gamma}\Big).
\end{array}
\end{equation}
\indent From Eqs.~\eqref{eq:qgpie1} and \eqref{eq:gappi},
we note that there exists a nonzero gap between $p_{\rm G}^{\rm PI}(\mathcal{E})$ and $p_{\rm L}^{\rm PI}(\mathcal{E})$,
\begin{equation}\label{eq:plpinl}
\begin{array}{c}
p_{\rm L}^{\rm PI}(\mathcal{E})
=\frac{1}{2}\Big(1+\frac{\sqrt{1+\gamma^{2}}}{1+\gamma}\Big)
<1=p_{\rm G}^{\rm PI}(\mathcal{E}).
\end{array}
\end{equation}
Thus, NLWE occurs in terms of optimal unambiguous discrimination when the post-measurement information about the prepared subensemble is available.\\
\begin{figure}[!t]
\centerline{\includegraphics*[bb=20 20 430 310,scale=0.8]{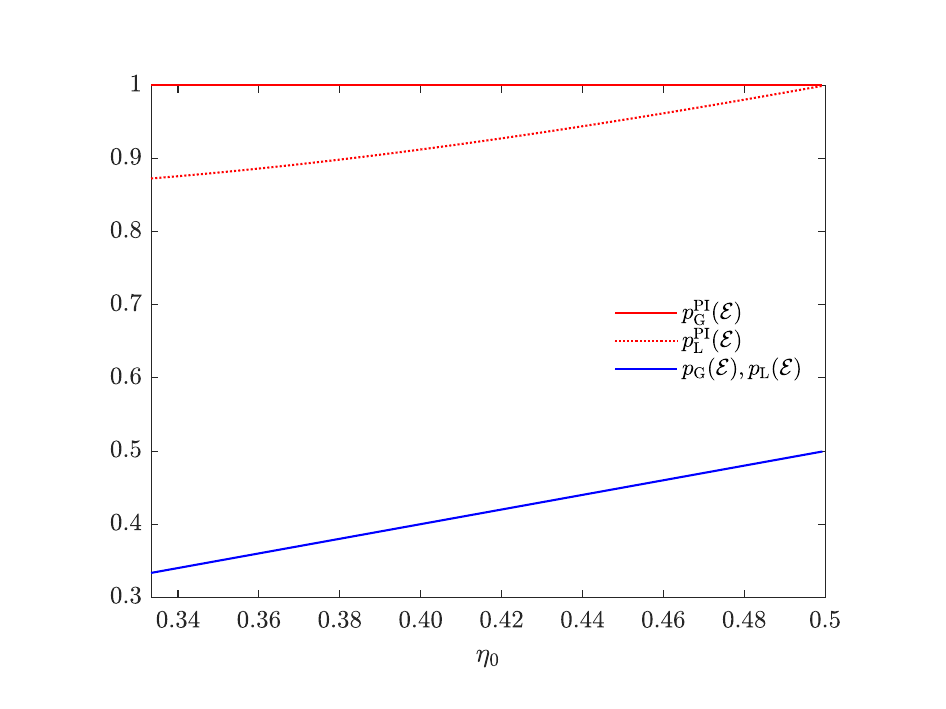}}
\caption{{\bf Unlocking NLWE by post-measurement information in terms of optimal unambiguous discrimination.} For all $\eta_{0}\in[\frac{1}{3},\frac{1}{2})$, 
$p_{\rm L}(\mathcal{E})$(blue) is equal to 
$p_{\rm G}(\mathcal{E})$(blue), 
but $p_{\rm L}^{\rm PI}(\mathcal{E})$(dashed red)
is less than $p_{\rm G}^{\rm PI}(\mathcal{E})$(solid red).}\label{fig:uloud}
\end{figure}
\indent Equation \eqref{eq:eqplpg} shows that NLWE in terms of optimal unambiguous discrimination does not occur  in discriminating the states of the ensemble $\mathcal{E}$ in Example~\ref{ex:unlock}, whereas Inequality \eqref{eq:plpinl} shows that NLWE occurs when post-measurement information is available.
Figure~\ref{fig:uloud} illustrates the relative order of $p_{\rm G}(\mathcal{E})$, $p_{\rm L}(\mathcal{E})$, $p_{\rm G}^{\rm PI}(\mathcal{E})$, and $p_{\rm L}^{\rm PI}(\mathcal{E})$ for the range of $\frac{1}{3}\leqslant\eta_{0}<\frac{1}{2}$.
\begin{theorem}\label{thm:uloud}
For optimal unambiguous discrimination of the ensemble $\mathcal{E}$ in Example~\ref{ex:unlock},
the post-measurement information about the prepared subensemble unlocks NLWE.
\end{theorem}
\section*{DISCUSSION}
We have shown that the post-measurement information about the prepared subensemble can lock or unlock NLWE in terms of optimal unambiguous discrimination. We have provided a quantum state ensemble consisting of four $2\otimes2$ non-entangled pure states (Example~\ref{ex:lock}) and shown 
the occurrence of NLWE in terms of optimal unambiguous discrimination with respect to the ensemble.
With the same state ensemble, we have further shown that the availability of post-measurement information about the prepared subensemble vanishes the occurrence of NLWE, thus locking NLWE in terms of optimal unambiguous discrimination by post-measurement information (Theorem~\ref{thm:loud}). Moreover, we have provided another quantum state ensemble consisting of four $2\otimes2$ non-entangled pure states (Example~\ref{ex:unlock}) and shown the non-occurrence of NLWE in terms of optimal unambiguous discrimination
with respect to the ensemble. With the same state ensemble, we have further shown the occurrence of NLWE in the optimal unambiguous discrimination with the post-measurement information about the prepared subensemble, thus unlocking NLWE in terms of optimal unambiguous discrimination by post-measurement information (Theorem~\ref{thm:uloud}).\\
\indent We remark that the two state ensembles of this paper can also be used to demonstrate locking and unlocking NLWE in terms of minimum-error discrimination\cite{ha20213}.
Thus, it is a natural future work to investigate locking and unlocking NLWE even in generalized state discrimination strategies such as an optimal discrimination with a fixed rate of inconclusive results\cite{chef1998,zhan1999,fiur2003,baga2012,herz2015}.\\
\indent Our results can also provide us with a useful application in quantum cryptography. Whereas the existing quantum data hiding and secret sharing schemes are based on orthogonal states\cite{terh20011,divi2002,egge2002,raha2015,wang20171},
our results can extend those schemes to improved ones using non-orthogonal states.
In Example~\ref{ex:lock}, the availability of the post-measurement information about the prepared subensemble makes the globally hidden information accessible locally.
On the other hand, in Example~\ref{ex:unlock}, the post-measurement information makes 
locally accessible information hidden locally but accessible globally. Finally, it is an interesting task to investigate if locking or unlocking NLWE
by the post-measurement information about the prepared subensemble
can depend on nonzero prior probabilities.\\
\section*{METHODS}\label{mtdsec}
In two-qubit (or $2\otimes2$) systems, a state
and a measurement are expressed by
a density operator and a positive operator-valued measure(POVM), respectively, acting on a two-party complex Hilbert space $\mathbb{C}^{2}\otimes\mathbb{C}^{2}$.
A density operator $\rho$ is a positive-semidefinite operator $\rho\succeq0$ with unit trace $\mathrm{Tr}\rho=1$ and a POVM $\{M_{i}\}_{i}$ is a set of positive-semidefinite operators $M_{i}\succeq0$ satisfying $\sum_{i}M_{i}=\mathbbm{1}$, where
$\mathbbm{1}$ is the identity operator on $\mathbb{C}^{2}\otimes\mathbb{C}^{2}$.
The probability of obtaining
the measurement outcome corresponding to $M_{i}$
is $\mathrm{Tr}(\rho M_{i})$ when $\{M_{i}\}_{i}$ is performed on a quantum system prepared with $\rho$.\\
\indent A positive-semidefinite operator is called \emph{separable}(or \emph{non-entangled}) if it is a sum of positive-semidefinite product operators;
otherwise, it is said to be \emph{entangled}.
Also, a POVM is called \emph{separable} if all elements are separable. In particular, a \emph{LOCC measurement} that can be realized by LOCC is a separable measurement\cite{chit20142}.
\subsection*{Optimal unambiguous discrimination}
Let us consider the unambiguous discrimination of the states in $\mathcal{E}$ of Eq.~\eqref{eq:ense} using a measurement $\{M_{i}\}_{i\in\overline{\Lambda}}$, 
where 
\begin{equation}
\overline{\Lambda}=\Lambda\cup\{?\}=\{0,1,+,-,?\}.
\end{equation}
For each $i\in\Lambda$, $M_{i}$ is to detect $\rho_{i}$,
and $M_{?}$ gives inconclusive results: ``I don't know what state is prepared.'' The measurement $\{M_{i}\}_{i\in\overline{\Lambda}}$ can be expressed as
\begin{equation}\label{eq:udm}
\begin{array}{c}
M_{i}=s_{i}|\tilde{\varphi}_{i}\rangle\!\langle\tilde{\varphi}_{i}|\ \ \forall i\in\Lambda,\
M_{?}=\mathbbm{1}-\sum_{j\in\Lambda}s_{j}|\tilde{\varphi}_{j}\rangle\!\langle\tilde{\varphi}_{j}|,
\end{array}
\end{equation}
where $\{s_{i}\}_{i\in\Lambda}$ is a non-negative number set and $\{|\tilde{\varphi}_{i}\rangle\}_{i\in\Lambda}$ is the set of reciprocal vectors
corresponding to $\{|\varphi_{i}\rangle\}_{i\in\Lambda}$ in Eq.~\eqref{eq:varphis} such that $\langle\varphi_{i}|\tilde{\varphi}_{j}\rangle=\delta_{ij}$\cite{pang2009}. 
We say a POVM $\{M_{i}\}_{i\in\overline{\Lambda}}$ is \emph{unambiguous} if it satisfies the error-free condition in Eq.~\eqref{eq:udm}.\\
\indent The \emph{optimal unambiguous discrimination of $\mathcal{E}$} is 
to minimize the probability of obtaining inconclusive results.
Equivalently, the optimal unambiguous discrimination of $\mathcal{E}$ is to maximize the average probability of 
unambiguously discriminating states in $\mathcal{E}$;
\begin{equation}\label{eq:pgeoud}
p_{\rm G}(\mathcal{E})=\max_{\rm Eq.\eqref{eq:udm}}
\sum_{i\in\Lambda}\eta_{i}\mathrm{Tr}(\rho_{i}M_{i})
\end{equation}
where the maximum is taken over all possible unambiguous measurements
satisfying the error-free condition in Eq.~\eqref{eq:udm}. 
It is known that an unambiguous measurement $\{M_{i}\}_{i\in\overline{\Lambda}}$
is optimal if and only if there is a positive-semidefinite operator $K$ satisfying the following condition\cite{elda20031,elda2004,naka20151,ha20212},
\begin{equation}\label{eq:kkt}
\begin{array}{c}
\langle\tilde{\varphi}_{i}|K|\tilde{\varphi}_{i}\rangle\geqslant\eta_{i}\ \forall i\in\Lambda,\
\mathrm{Tr}[M_{i}(K-\eta_{i}\rho_{i})]=0\ \forall i\in\Lambda,\
\mathrm{Tr}(M_{?}K)=0.
\end{array}
\end{equation}
In this case, we have
\begin{equation}
p_{\rm G}(\mathcal{E})=\sum_{i\in\Lambda}\eta_{i}\mathrm{Tr}(\rho_{i}M_{i})={\rm Tr}K
\end{equation}
if an unambiguous measurement $\{M_{i}\}_{i\in\overline{\Lambda}}$ and
a positive-semidefinite operator $K$ satisfy Condition \eqref{eq:kkt}\cite{elda20031,elda2004,naka20151,ha20212}.\\
\indent When the available measurements are restricted to unambiguous LOCC measurements, we denote the maximum success probability by
\begin{equation}\label{eq:qlelocc}
p_{\rm L}(\mathcal{E})=\max_{\substack{\rm Eq.\eqref{eq:udm}\\{\rm LOCC}}}\sum_{i\in\Lambda}\eta_{i}\mathrm{Tr}(\rho_{i}M_{i}).
\end{equation}
\indent In the following lemma, we provide an upper bound of $p_{\rm L}(\mathcal{E})$. 
\begin{lemma}\label{lem:qletrh}
If $H$ is a positive-semidefinite operator satisfying
\begin{equation}\label{eq:pth}
\langle\tilde{\varphi}_{i}|H|\tilde{\varphi}_{i}\rangle\geqslant \eta_{i}
\end{equation}
for all reciprocal vectors $|\tilde{\varphi}_{i}\rangle$ that is a product vector, then ${\rm Tr}H$ is an upper bound of $p_{\rm L}(\mathcal{E})$.
\end{lemma}
\begin{proof}
Let us suppose that $\{M_{i}\}_{i\in\overline{\Lambda}}$ is an unambiguous  LOCC measurement and $\chi$ is the set of all $i\in\Lambda$ such that $|\tilde{\varphi}_{i}\rangle$ is a product vector.
Since every LOCC measurement is separable,
$M_{i}$ is separable for all $i\in\overline{\Lambda}$.
For all $i\in\Lambda$ with $i\notin\chi$, $M_{i}=0$ 
because $M_{i}$ is proportional to entangled $|\tilde{\varphi}_{i}\rangle\!\langle\tilde{\varphi}_{i}|$.
Thus, the success probability is
\begin{equation}\label{eq:splem}
\begin{array}{c}
\sum_{i\in\chi}\eta_{i}{\rm Tr}(\rho_{i}M_{i})
\leqslant\sum_{i\in\chi}\eta_{i}{\rm Tr}(\rho_{i}M_{i})
+\sum_{i\in\chi}{\rm Tr}[(H-\eta_{i}\rho_{i})M_{i}]+{\rm Tr}(HM_{?})
={\rm Tr}H,
\end{array}
\end{equation}
where the inequality is due to the assumption of Inequality~\eqref{eq:pth} and the positive-semidefiniteness of $H$ and $M_{?}$, and
the equality is from  $M_{?}=\mathbbm{1}-\sum_{i\in\chi}M_{i}$.
As Inequality~\eqref{eq:splem} is true for any unambiguous LOCC measurement $\{M_{i}\}_{i\in\overline{\Lambda}}$, 
${\rm Tr}H$ is an upper bound of $p_{\rm L}(\mathcal{E})$.
\end{proof}
\subsection*{Optimal unambiguous discrimination with post-measurement information}
\indent Let us consider 
the situation of unambiguously discriminating the two-qubit states of $\mathcal{E}$ in Eq.~\eqref{eq:ense} when the classical information $b\in\{0,1\}$ about the prepared subensemble $\mathcal{E}_{b}$ defined in Eq.~\eqref{eq:suben}
is given after performing a measurement.
We use a POVM $\{M_{\vec{\omega}}\}_{\vec{\omega}\in\Omega}$
to unambiguously discriminate the states of $\mathcal{E}$ in Eq.~\eqref{eq:ense}, where $\Omega$ is the Cartesian product of two outcome sets $\mathsf{A}_{0}\cup\{?\}$ and $\mathsf{A}_{1}\cup\{?\}$ with inconclusive results,
\begin{equation}
\begin{array}{rcl}
\Omega&=&\{(\omega_{0},\omega_{1})\,|\,\omega_{0}\in\mathsf{A}_{0}\cup\{?\},\,
\omega_{1}\in\mathsf{A}_{1}\cup\{?\}\}\\[1mm]
&=&\{(0,+),(0,-),(1,+),(1,-),(0,?),(1,?),(?,+),(?,-),(?,?)\}.
\end{array}
\end{equation}
For each $(\omega_{0},\omega_{1})\in\Omega$,
$M_{(\omega_{0},\omega_{1})}$ detects a state in $\mathcal{E}$ unambiguously or gives inconclusive results 
depending on post-measurement information $b\in\{0,1\}$.
If $\omega_{b}\neq?$, the state $\rho_{\omega_{b}}$ is detected unambiguously, that is, the POVM $\{M_{\vec{\omega}}\}_{\vec{\omega}\in\Omega}$ satisfies
\begin{equation}\label{eq:udmpi}
\begin{array}{lll}
\mathrm{Tr}[\rho_{-}M_{(i,+)}]
=\mathrm{Tr}[\rho_{+}M_{(i,-)}]
=0& \forall i\,\in\mathsf{A}_{0}\cup\{?\},\\
\mathrm{Tr}[\rho_{\,1}M_{(0,\,j)}]
=\mathrm{Tr}[\rho_{\,0\,}M_{(1,\,j)}]
=0& \forall j\in\mathsf{A}_{1}\cup\{?\}.
\end{array}
\end{equation}
However, if $\omega_{b}=?$, inconclusive results are obtained.
We say that a POVM $\{M_{\vec{\omega}}\}_{\vec{\omega}\in\Omega}$ is \emph{unambiguous} if it satisfies the error-free condition in Eq.~\eqref{eq:udmpi}.\\
\indent The \emph{optimal unambiguous discrimination of $\mathcal{E}$ with post-measurement information} is
to minimize the probability of obtaining 
inconclusive results.
Equivalently, the optimal unambiguous discrimination of $\mathcal{E}$ with post-measurement information is to maximize the average probability of 
unambiguously discriminating states where the optimal success probability is defined as
\begin{equation}\label{eq:qgpie}
p_{\rm G}^{\rm PI}(\mathcal{E})=
\max_{\rm Eq.\eqref{eq:udmpi}}
\sum_{b\in\{0,1\}}\sum_{i\in\mathsf{A}_{b}}
\eta_{i}\mathrm{Tr}\Big[\rho_{i} \sum_{\substack{\vec{\omega}\in\Omega\\\omega_{b}=i}} M_{\vec{\omega}}\Big]
\end{equation}
over all possible unambiguous measurements in Eq.~\eqref{eq:udmpi}.\\
\indent Rather surprisingly, some non-orthogonal states can be perfectly discriminated when the post-measurement information about the prepared subensemble is available\cite{akib2019}, that is, $p_{\rm G}^{\rm PI}(\mathcal{E})=1$. More precisely, for a state ensemble $\mathcal{E}$ that
consists of two subensembles $\mathcal{E}_{0}$ and $\mathcal{E}_{1}$ with two pure states, $p_{\rm G}^{\rm PI}(\mathcal{E})=1$ if and only if
\begin{equation}\label{eq:pfdc}
(1-G_{0+})(1-G_{1-})+(1-G_{0-})(1-G_{1+})-2\sqrt{G_{0+}G_{0-}G_{1+}G_{1-}}\geqslant1,
\end{equation}
where
\begin{equation}
G_{ij}=\mathrm{Tr}(\rho_{i}\rho_{j}),~i\in\mathsf{A}_{0},~j\in\mathsf{A}_{1}.
\end{equation}
\indent When the available measurements are limited to unambiguous LOCC measurements,
we denote the maximum success probability by
\begin{eqnarray}
p_{\rm L}^{\rm PI}(\mathcal{E})&=&
\max_{\substack{{\rm Eq.\eqref{eq:udmpi}}\\{\rm LOCC}}}
\sum_{b\in\{0,1\}}\sum_{i\in\mathsf{A}_{b}}
\eta_{i}\mathrm{Tr}\Big[\rho_{i} \sum_{\substack{\vec{\omega}\in\Omega\\\omega_{b}=i}} M_{\vec{\omega}}\Big].\label{eq:qlpie}
\end{eqnarray}
We note that $p_{\rm L}^{\rm PI}(\mathcal{E})$ in Eq.~\eqref{eq:qlpie} can also be rewritten as
\begin{eqnarray}
p_{\rm L}^{\rm PI}(\mathcal{E})
=\max_{\substack{{\rm Eq.\eqref{eq:udmpi}}\\{\rm LOCC}}}\Bigg[\sum_{\vec{\omega}\in\mathsf{A}_{0}\times\mathsf{A}_{1}}\tilde{\eta}_{\vec{\omega}}
\mathrm{Tr}(\tilde{\rho}_{\vec{\omega}}M_{\vec{\omega}})
+\sum_{i\in\mathsf{A}_{0}}\eta_{i}\mathrm{Tr}(\rho_{i}M_{(i,?)})+\sum_{j\in\mathsf{A}_{1}}\eta_{j}\mathrm{Tr}(\rho_{j}M_{(?,j)})
\Bigg],
\end{eqnarray}
where 
\begin{equation}\label{eq:trhow}
\begin{array}{rcl}
\tilde{\eta}_{\vec{\omega}}
=\frac{1}{2}\sum_{b\in\{0,1\}}\eta_{w_{b}},\ 
\tilde{\rho}_{\vec{\omega}}
=\frac{\sum_{b\in\{0,1\}}\eta_{w_{b}}\rho_{\omega_{b}}}{\sum_{b'\in\{0,1\}}\eta_{w_{b'}}}.
\end{array}
\end{equation}
We also note that 
Inequality~\eqref{eq:pfdc} is a necessary but not sufficient condition for $p_{\rm L}^{\rm PI}(\mathcal{E})=1$ because
$p_{\rm L}^{\rm PI}(\mathcal{E})=1$ implies $p_{\rm G}^{\rm PI}(\mathcal{E})=1$ but not vice versa.\\
\indent For an upper bound of $p_{\rm L}^{\rm PI}(\mathcal{E})$, let us consider the following quantity,
\begin{equation}\label{eq:plelocc}
p_{\rm guess}(\mathcal{E})=\max_{\substack{\{M_{i}\}_{i\in\Lambda}:\\{\rm POVM}}}\sum_{i\in\Lambda}\eta_{i}\mathrm{Tr}(\rho_{i}M_{i}),
\end{equation}
which is the maximum average probability of correct guessing the prepared state when the available measurements are limited to LOCC measurements without inconclusive results\cite{hels1976,hole1979,yuen1975,bae2013}. The following lemma shows that $p_{\rm guess}(\mathcal{E})$ can be used as an upper bound of $p_{\rm L}^{\rm PI}(\mathcal{E})$.
\begin{lemma}\label{lem:qlpieple}
For each $(\omega_{0},\omega_{1})\in\mathsf{A}_{0}\times\mathsf{A}_{1}$, if
there is no nonzero product vector $|v\rangle$ satisfying
\begin{equation}\label{eq:npvc}
\begin{array}{c}
\langle\varphi_{i}|v\rangle\neq0\
\forall i\in\{\omega_{0},\omega_{1}\},\
\langle\varphi_{j}|v\rangle=0\ 
\forall j\in\Lambda\setminus\{\omega_{0},\omega_{1}\},
\end{array}
\end{equation} 
then $p_{\rm guess}(\mathcal{E})$ is an upper bound of $p_{\rm L}^{\rm PI}(\mathcal{E})$.
\end{lemma}
\begin{proof}
The assumption in \eqref{eq:npvc} implies that
for each $\vec{\omega}=(\omega_{0},\omega_{1})\!\in\!\mathsf{A}_{0}\!\times\!\mathsf{A}_{1}$, there does not exist any nonzero separable $M_{\vec{\omega}}\succeq0$ that unambiguously detects the state $\rho_{\omega_{0}}$ or $\rho_{\omega_{1}}$ depending on post-measurement information $b=0$ or $1$, respectively.
Then, the term $\sum_{\vec{\omega}\in\mathsf{A}_{0}\times\mathsf{A}_{1}}\tilde{\eta}_{\vec{\omega}}\mathrm{Tr}(\tilde{\rho}_{\vec{\omega}}M_{\vec{\omega}})$ in Eq.~\eqref{eq:qlpie} disappears. Thus, we have
\begin{eqnarray}
p_{\rm L}^{\rm PI}(\mathcal{E})
&=&\max_{\substack{{\rm Eq.\eqref{eq:udmpi}}\\{\rm LOCC}}}\Bigg[\sum_{i\in\mathsf{A}_{0}}\eta_{i}\mathrm{Tr}(\rho_{i}M_{(i,?)})
+\sum_{j\in\mathsf{A}_{1}}\eta_{j}\mathrm{Tr}(\rho_{j}M_{(?,j)})
\Bigg]\nonumber\\
&\leq&\max_{\substack{\{M_{i}\}_{i\in\Lambda}:\\{\rm LOCC}}}\Bigg[\sum_{i\in\mathsf{A}_{0}}\eta_{i}\mathrm{Tr}(\rho_{i}M_{i})
+\sum_{j\in\mathsf{A}_{1}}\eta_{j}\mathrm{Tr}(\rho_{j}M_{j})
\Bigg]=p_{\rm guess}(\mathcal{E}),
\end{eqnarray}
where the inequality is from the fact that
$p_{\rm guess}(\mathcal{E})$ is the maximum obtained from measurements without any constraint, 
whereas $p_{\rm L}^{\rm PI}(\mathcal{E})$ is the maximum obtained from 
unambiguous LOCC measurements.
\end{proof}
\section*{Acknowledgements}
This work was supported by Quantum Computing Technology Development Program(NRF-2020M3E4A1080088) through the National Research Foundation of Korea(NRF) grant funded by the Korea government(Ministry of Science and ICT).

\end{document}